\documentclass[runningheads,a4paper]{llncs}

\usepackage{amssymb,amsmath,tikz}
\usepackage{graphicx}
\usetikzlibrary{patterns}

\usepackage{url}
\usepackage{enumerate}
\urldef{\maila}\path|cel@mit.edu|
\newcommand{\keywords}[1]{\par\addvspace\baselineskip
\noindent\keywordname\enspace\ignorespaces#1}

\makeatletter
\newcommand{\text@hyphens}{\mathcode`\-=`\-\relax}
\newcommand{\id}[1]{\ensuremath{\mathit{\text@hyphens#1}}}
\makeatother

\begin{document}

\mainmatter  

\title{On the Efficiency of Localized Work Stealing}

\titlerunning{On the Efficiency of Localized Work Stealing}

%
%
\author{Warut Suksompong\inst{1}\and Charles E. Leiserson\inst{2}\and Tao B. Schardl\inst{2}%
}
\authorrunning{W. Suksompong, C. E. Leiserson, T. B. Schardl}

\institute{Department of Computer Science, Stanford University\\
353 Serra Mall, Stanford, CA 94305, USA\\
\email{warut@cs.stanford.edu}
\and
MIT Computer Science and Artificial Intelligence Laboratory\\
32 Vassar St, Cambridge, MA 02139, USA\\
\email{\{cel,neboat\}@mit.edu}
}

%
%

\maketitle

\begin{abstract}
This paper investigates a variant of the work-stealing algorithm that we call the \textit{localized work-stealing algorithm}. The intuition behind this variant is that because of locality, processors can benefit from working on their own work. Consequently, when a processor is free, it makes a steal attempt to get back its own work. We call this type of steal a \textit{steal-back}. We show that the expected running time of the algorithm is $T_1/P+O(T_\infty P)$, and that under the ``even distribution of free agents assumption'', the expected running time of the algorithm is $T_1/P+O(T_\infty\lg P)$. In addition, we obtain another running-time bound based on ratios between the sizes of serial tasks in the computation. If $M$ denotes the maximum ratio between the largest and the smallest serial tasks of a processor after removing a total of $O(P)$ serial tasks across all processors from consideration, then the expected running time of the algorithm is $T_1/P+O(T_\infty M)$.
\end{abstract}

\section{Introduction}

Work stealing is an efficient and popular paradigm for scheduling multithreaded computations. While its practical benefits have been known for decades \cite{BurtonSl81,Halstead84} and several researchers have found applications of the paradigm \cite{AroraBlPl98,DinanLaSaKrNi09,KarpZh93,LeisersonScSu15}, Blumofe and Leiserson \cite{BlumofeLe99} were the first to give a theoretical analysis of work stealing. Their scheduler executes a fully strict (i.e., well-structured) multithreaded computations on $P$ processors within an expected time of $T_1/P+O(T_\infty)$, where $T_1$ is the minimum serial execution time of the multithreaded computation (the \textit{work} of the computation) and $T_\infty$ is the minimum execution time with an infinite number of processors (the \textit{span} of the computation.) 

In multithreaded computations, it sometimes occurs that a processor performs some computations and stores the results in its cache. Therefore, a work-stealing algorithm could potentially benefit from exploiting locality, i.e., having processors work on their own work as much as possible. Indeed, an experiment by Acar et al. \cite{AcarBlBl00} demonstrates that exploiting locality can improve the performance of the work-stealing algorithm by up to 80\%. Similarly, Guo et al. \cite{GuoZhCa10} found that locality-aware scheduling can achieve up to 2.6$\times$ speedup over locality-oblivious scheduling. In addition, work-stealing strategies that exploit locality have been proposed. Hierarchical work stealing, considered by Min et al. \cite{MinIaYe11} and Quintin and Wagner \cite{QuintinWa10}, contains mechanisms that find the nearest victim thread to preserve locality and determine the amount of work to steal based on the locality of the victim thread. More recently, Paudel et al. \cite{PaudelTaAm13} explored a selection of tasks based on the application-level task locality rather than hardware memory topology.

In this paper, we investigate a variant of the work-stealing algorithm that we call the \textit{localized work-stealing algorithm}. In the localized work-stealing algorithm, when a processor is free, it makes a steal attempt to get back its own work. We call this type of steal a \textit{steal-back}. We show that the expected running time of the algorithm is $T_1/P+O(T_\infty P)$, and that under the ``even distribution of free agents assumption'', the expected running time of the algorithm is $T_1/P+O(T_\infty\lg P)$. In addition, we obtain another running-time bound based on ratios between the sizes of serial tasks in the computation. If $M$ denotes the maximum ratio between the largest and the smallest serial tasks of a processor after removing a total of $O(P)$ serial tasks across all processors from consideration, then the expected running time of the algorithm is $T_1/P+O(T_\infty M)$.




This paper is organized as follows. Section \ref{sec:localizedsetting} introduces the setting that we consider throughout the paper. Section \ref{sec:delayseq} analyzes the localized work-stealing algorithm using the delay-sequence argument. Section \ref{sec:amortization} analyzes the algorithm using amortization arguments. Section \ref{sec:localizedvariants} considers variants of the localized work-stealing algorithm. Finally, Section \ref{sec:conclusion} concludes and suggests directions for future work.

\section{Localized Work-Stealing Algorithm}
\label{sec:localizedsetting}


Consider a setting with $P$ processors. Each processor owns some pieces of work, which we call \textit{serial tasks}. Each serial task takes a positive integer amount of time to complete, which we define as the \textit{size} of the serial task.  We assume that different serial tasks can be done in parallel and model the work of each processor as a binary tree whose leaves are the serial tasks of that processor. The trees are balanced in terms of the number of serial tasks on each branch, but the order in which the tasks occur in the binary tree is assumed to be given to us. We then connect the $P$ roots as a binary tree of height $\lg P$, so that we obtain a larger binary tree whose leaves are the serial tasks of all processors. 

As usual, we define $T_1$ as the work of the computation, and $T_\infty$ as the span of the computation. The span $T_\infty$ corresponds to the height of the aforementioned larger binary tree plus the size of the largest serial task. In addition, we define $T_{\infty}'$ as the height of the tree not including the part connecting the $P$ processors of height $\lg P$ at the top or the serial tasks at the bottom. Since $T_{\infty}'$ corresponds to a smaller part of the tree than $T_{\infty}$, we have $T_{\infty}'<T_{\infty}$.

The randomized work-stealing algorithm \cite{BlumofeLe99} suggests that whenever a processor is free, it should ``steal'' randomly from a processor that still has work left to do. In our model, stealing means taking away one of the two main branches of the tree corresponding to a particular processor, in particular, the branch that the processor is not working on. The randomized work-stealing algorithm performs $O(P(T_\infty+\lg(1/\epsilon)))$ steal attempts with probability at least $1-\epsilon$, and the execution time is $T_1/P+O(T_\infty+\lg P+\lg(1/\epsilon))$ with probability at least $1-\epsilon$.

This paper investigates a localized variant of the work-stealing algorithm. In this variant, whenever a processor is free, it first checks whether some other processors are working on its work. If so, it ``steals back'' randomly only from these processors. Otherwise, it steals randomly as usual. We call the two types of steal a \textit{general steal} and a \textit{steal-back}. The intuition behind this variant is that sometimes a processor performs some computations and stores the results in its cache. Therefore, a work-stealing algorithm could potentially benefit from exploiting locality, i.e., having processors work on their own work as much as possible.

We make a simplifying assumption that each processor maintains a list of the other processors that are working on its work. When a general steal occurs, the stealer adds its name to the list of the owner of the serial task that it has just stolen (not necessarily the same as the processor from which it has just stolen.) For example, if processor $P_1$ steals a serial task owned by processor $P_2$ from processor $P_3$, then $P_1$ adds its name to the $P_2$'s list (and not $P_3$'s list.) When a steal-back is unsuccessful, the owner removes the name of the target processor from its list, since the target processor has finished the owner's work.

An example of an execution of localized work-stealing algorithm can be found in \cite{Suksompong14}. We assume that the overhead for maintaining the list and dealing with contention for steal-backs is constant. This assumption is reasonable because adding (and later removing) the name of a processor to a list is done when a general steal occurs, and hence can be amortized with general steals. Randomizing a processor from the list to steal back from takes constant time. When multiple processors attempt to steal back from the same processor simultaneously, we allow an arbitrary processor to succeed and the remaining processors to fail, and hence do not require extra processing time.

\section{Delay-Sequence Argument}
\label{sec:delayseq}

In this section, we apply the delay-sequence argument to establish an upper bound on the running time of the localized work-stealing algorithm. The delay-sequence argument is used in \cite{BlumofeLe99} to show that the randomized work-stealing algorithm performs $O(P(T_\infty+\lg(1/\epsilon)))$ steal attempts with probability at least $1-\epsilon$. We show that under the ``even distribution of free agents assumption'', the expected running time of the algorithm is $T_1/P+O(T_\infty\lg P)$. We also show a weaker bound that without the assumption, the expected running time of the algorithm is $T_1/P+O(T_\infty P)$.

Since the amount of work done in a computation is always given by $T_1$, independent of the sequence of steals, we focus on estimating the number of steals. We start with the following definition. 

\begin{definition}
The \textit{even distribution of free agents assumption} is the assumption that when there are $k$ \textit{owners} left (and thus $P-k$ \textit{free agents}), the $P-k$ free agents are evenly distributed working on the work of the $k$ owners. That is, each owner has $P/k$ processors working on its work.
\end{definition}

While this assumption might not hold in the localized work-stealing algorithm as presented here, it is intuitively more likely to hold under the hashing modification presented in Section \ref{sec:localizedvariants}. When the assumption does not hold, we obtain a weaker bound as given in Theorem \ref{thm:distweakbound}.

Before we begin the proof of our theorem, we briefly summarize the delay-sequence argument as used by Blumofe and Leiserson \cite{BlumofeLe99}. The intuition behind the delay-sequence argument is that in a random process in which multiple paths of the process occur simultaneously, such as work stealing, there exists some path that finishes last. We call this path the \textit{critical path}. The goal of the delay-sequence argument is to show that it is unlikely that the process takes a long time to finish by showing that it is unlikely that the critical path takes a long time to finish. To this end, we break down the process into \textit{rounds}. We define a round so that in each round, there is a constant probability that the critical path is shortened. (In the case of work stealing, this means there exists a steal on the critical path.) This will allow us to conclude that there are not too many rounds, and consequently not too many steals in the process.

\begin{theorem}
\label{thm:evendist}
With the even distribution of free agents assumption, the number of steal attempts is $O(P\lg P(T_\infty+\lg(P/\epsilon)))$ with probability at least $1-\epsilon$, and the expected number of steal attempts is $O(P\lg PT_\infty)$. 
\end{theorem}

\begin{proof}
Consider any processor. At timestep $t$, let $S^t$ denote the number of general steals occurring at that timestep, and let $X^t$ be the random variable
\[
    X^t= 
\begin{cases}
    1/P^t,& \text{if the processor can steal back from } P^t \text{ other processors};\\
    0,              & \text{if the processor is working.}
\end{cases}
\]

We define a \textit{round} to be a consecutive number of timesteps $t$ such that 
\[\sum_t\left(S^t+PX^t\right)\geq P,\] 
and such that this inequality is not satisfied if we remove the last timestep from the round. Note that this condition is analogous to the condition of a round in \cite{BlumofeLe99}, where the number of steals is between $3P$ and $4P$. Here we have the term $S^t$ corresponding to general steals and the term $PX^t$ corresponding to steal-backs. 

We define the \textit{critical path} of the processor to be the path from the top of its binary tree to the serial task of the processor whose execution finishes last. We show that any round has a probability of at least $1-1/e$ of reducing the length of the critical path.

We compute the probability that a round does not reduce the length of the critical path. Each general steal has a probability of at least $1/P$ of stealing off the critical path and thus reducing its length. Each steal-back by the processor has a probability of $1/P^t$ of reducing the length of the critical path. At timestep $t$, the probability of not reducing the length of the critical path is therefore
\[ \left(1-\frac{1}{P}\right)^{S^t}\left(1-X^t\right)\leq e^{-\frac{S^t}{P}-X^t}, \] where we used the inequality $1+x\leq e^x$ for all real numbers $x$. Therefore, the probability of not reducing the length of the critical path during the whole round is at most 
\[ \prod_t e^{-\frac{S^t}{P}-X^t} = e^{-\sum_t\left(\frac{S^t}{P}+X^t\right)} \leq e^{-1}. \]

Note that this bound remains true even when there are concurrent thieves, since we are concerned with the probability that in a given round the length of the critical path is not reduced. If there are concurrent thieves trying to make a steal on the critical path, one of them will be successful, and the other unsuccessful thieves do not play a role in our analysis.

With this definition of a round, we can now apply the delay-sequence argument as in \cite{BlumofeLe99}. Note that in a single timestep $t$, we have $S^t\leq P$ and $PX^t\leq P$. Consequently, in every round, we have $P\leq\sum_t\left(S^t+PX^t\right)\leq 3P.$ 

Suppose that over the course of the whole execution, we have $\sum_t\left(S^t+PX^t\right)\geq 3PR$, where $R=cT_\infty + \lg(1/\epsilon)$ for some sufficiently large constant $c$. Then there must be at least $R$ rounds. Since each round has a probability of at most $e^{-1}$ of not reducing the length of the critical path, the delay-sequence argument yields that the probability that $\sum_t\left(S^t+PX^t\right)\geq 3PR=\Theta(P(T_\infty+\lg(1/\epsilon)))$ is at most $\epsilon$.

We apply the same argument to every processor. Suppose without loss of generality that processor 1's work is completed first, then processor 2's work, and so on, up to processor $P$'s work. Let $S_i$ denote the number of general steals up to the timestep when processor $i$'s work is completed, and let $X_i^t$ denote the value of the random variable $X^t$ corresponding to processor $i$. In particular, $S_P$ is the total number of general steals during the execution, which we also denote by $S$. We have 
\[\text{Pr}\left[S_i+\sum_tPX_i^t\geq \Theta(P(T_\infty+\lg(1/\epsilon)))\right]\leq \epsilon.\]

Now we use our even distribution of free agents assumption. This means that when processor $i$ steals back, there are at most $(i-1)/(P-i+1)$ processors working on its work. Hence $X_i^t\geq (P-i+1)/(i-1)$ whenever $X_i^t\neq 0$. Letting $W_i$ be the number of steal-backs performed by processor $i$, we have
\[\text{Pr}\left[S_i+\frac{P(P-i+1)}{i-1}W_i\geq \Theta(P(T_\infty+\lg(1/\epsilon)))\right]\leq \epsilon.\]
For processor $2\leq i\leq P-1$, this says
\[\text{Pr}\left[\dfrac{i-1}{P(P-i+1)}S_i+W_i\geq \Theta\left(\dfrac{i-1}{P-i+1}(T_\infty+\lg(1/\epsilon))\right)\right]\leq \epsilon.\]
In particular, we have 
\[\text{Pr}\left[W_i\geq \Theta\left(\dfrac{i-1}{P-i+1}(T_\infty+\lg(1/\epsilon))\right)\right]\leq \epsilon.\]
For processor $P$, we have 
\[\text{Pr}\left[S+\frac{P}{P-1}W_P\geq \Theta(P(T_\infty+\lg(1/\epsilon)))\right]\leq \epsilon.\]
Since $P\geq P-1$, we have 
\[\text{Pr}\left[S+W_P\geq \Theta(P(T_\infty+\lg(1/\epsilon)))\right]\leq \epsilon.\]
Since $\sum_{i=2}^{P-1}\frac{i-1}{P-i+1}$ grows as $P\lg P$, adding up the estimates for each of the $P$ processors and using the union bound, we have 
\[\text{Pr}\left[S+\sum_{i=1}^PW_i\geq \Theta(P\lg P(T_\infty+\lg(1/\epsilon)))\right]\leq P\epsilon.\]
Substituting $\epsilon$ with $\epsilon/P$ yields the desired bound. 

Since the tail of the distribution decreases exponentially, the expectation bound follows.
\end{proof}

The bound on the execution time follows from Theorem \ref{thm:evendist}.

\begin{theorem}
With the even distribution of free agents assumption, the expected running time, including scheduling overhead, is $T_1/P+O(T_\infty\lg P)$. Moreover, for any $\epsilon>0$, with probability at least $1-\epsilon$, the execution time on $P$ processors is $T_1/P+O(\lg P(T_\infty+\lg(P/\epsilon)))$.
\end{theorem}

\begin{proof}
The amount of work is $T_1$, and Theorem \ref{thm:evendist} gives a bound on the number of steal attempts. We add up the two quantities and divide by $P$ to complete the proof.
\end{proof}


Without the even distribution of free agents assumption, we obtain a weaker bound, as the following theorem shows.

\begin{theorem}
\label{thm:distweakbound}
The number of steal attempts is $O(P^2(T_\infty+\lg(P/\epsilon)))$ with probability at least $1-\epsilon$.
\end{theorem}

\begin{proof}
We apply a similar analysis using the delay-sequence argument as in Theorem \ref{thm:evendist}. The difference is that here we have $X_i^t\geq 1/P$ instead of $X_i^t\geq (P-i+1)/(i-1)$. Hence, instead of 
\[\text{Pr}\left[S_i+\frac{P(P-i+1)}{i-1}W_i\geq \Theta(P(T_\infty+\lg(1/\epsilon)))\right]\leq \epsilon,\]
we have
\[\text{Pr}\left[S_i+W_i\geq \Theta(P(T_\infty+\lg(1/\epsilon)))\right]\leq \epsilon.\]
The rest of the analysis proceeds using the union bound as in Theorem \ref{thm:evendist}.
\end{proof}

Again, the bound on the execution time follows from Theorem \ref{thm:distweakbound}.

\begin{theorem}
The expected running time of the localized work-stealing algorithm, including scheduling overhead, is $T_1/P+O(T_\infty P)$. Moreover, for any $\epsilon>0$, with probability at least $1-\epsilon$, the execution time on $P$ processors is $T_1/P+O(P(T_\infty+\lg(P/\epsilon)))$.
\end{theorem}

\begin{proof}
The amount of work is $T_1$, and Theorem \ref{thm:distweakbound} gives a bound on the number of steal attempts. We add up the two quantities and divide by $P$ to complete the proof.
\end{proof}

\begin{remark}
In the delay-sequence argument, it is not sufficient to consider the critical path of only one processor (e.g., the processor that finishes last.)

For example, suppose that there are 3 processors, $P_1,P_2$, and $P_3$. $P_1$ owns 50 serial tasks of size 1 and 1 serial task of size 100, $P_2$ owns 1 serial task of size 1 and 1 serial task of size 1000, and $P_3$ owns no serial task. At the beginning of the execution, $P_3$ has a probability of 1/2 of stealing from $P_1$. If it steals from $P_1$ and gets stuck with the serial task of size 100, $P_1$ will perform several steal-backs from $P_3$, while the critical path is on $P_2$'s subtree. 

Hence, the steal-backs by $P_1$ do not contribute toward reducing the length of the critical path. 
\end{remark}

We briefly discuss the scalability of our localized work-stealing strategy. The bound $T_P \leq T_1/P + O(T_{\infty})$ provided by Blumofe and Leiserson \cite{BlumofeLe99} means that when $P \ll T_1/T_{\infty}$, we achieve linear speedup, i.e., $T_P \approx T_1/P$. Indeed, when $P \ll T_1/T_{\infty}$, we have that $T_{\infty}\ll T_1/P$, which implies that the term $T_1/P$ is the dominant term in the sum $T_1/P + O(T_{\infty})$. On the other hand, for our bound of $T_P \leq T_1/P + O(T_{\infty} P)$, when $P \ll \sqrt{T_1/T_{\infty}}$, we have that $T_{\infty}P\ll T_1/P$, and hence the term $T_1/P$ dominates in the sum $T_1/P + O(T_{\infty} P)$. As a result, we achieve linear speedup in localized work stealing when $P \ll \sqrt{T_1/T_{\infty}}$. In other words, we have square-rooted the effective parallelism.  Thus the application scales, but not as readily as in vanilla randomized work stealing.

\section{Amortization Analysis}
\label{sec:amortization}

In this section, we apply amortization arguments to obtain bounds on the running time of the localized work-stealing algorithm. We show that if $M$ denotes the maximum ratio between the largest and the smallest serial tasks of a processor after removing a total of $O(P)$ serial tasks across all processors from consideration, then the expected running time of the algorithm is $T_1/P+O(T_\infty M)$.

We begin with a simple bound on the number of steal-backs.

\begin{theorem}
The number of steal-backs is at most $T_1+O(PT_\infty)$ with high probability.
\end{theorem}

\begin{proof}
Every successful steal-back can be amortized by the work done by the stealer in the timestep following the steal-back. Every unsuccessful steal-back can be amortized by a general steal. Indeed, recall our assumption that after each unsuccessful steal-back, the target processor is removed from the owner's list. Hence each general steal can generate at most one unsuccessful steal-back. Since there are at most $O(PT_\infty)$ general steals with high probability, we obtain the desired bound.
\end{proof}

The next theorem amortizes each steal-back against general steals, using the height of the tree to estimate the number of general steals.

\begin{theorem}
\label{steal-back}
Let $N$ denote the number of general steals in the computation, and let $T_\infty'$ denote the height of the tree not including the part connecting the $P$ processors of height $\lg P$ at the top or the serial tasks at the bottom. (In particular, $T_\infty'<T_\infty.$) Then there are at most $T_\infty'N$ steal-back attempts.
\end{theorem}

\begin{proof}
Suppose that a processor $P_i$ steals back from another processor $P_j$. This means that earlier, $P_j$ performed a general steal on $P_i$ which resulted in this steal-back. We amortize the steal-back against the general steal. Each general steal generates at most $T_\infty'$ steal-backs (or $T_\infty'+1$, to be more precise, since there can be an unsuccessful steal-back after $P_j$ completed all of $P_i$'s work and $P_i$ erased $P_j$'s name from its list.) Since there are $N$ general steals in our computation, there are at most $T_\infty'N$ steal-back attempts.

After $P_j$ performed the general steal on $P_i$, it is possible some other processor $P_k$ makes a general steal on $P_j$. This does not hurt our analysis. When $P_i$ steals back from $P_k$, we amortize the steal-back against the general steal that $P_k$ makes on $P_j$, not the general steal that $P_j$ makes on $P_i$.
\end{proof}

Since there are at most $O(PT_\infty)$ general steals with high probability, Theorem \ref{steal-back} shows that there are at most $O(T_\infty'PT_\infty)$ steals in total with high probability. 

The next theorem again amortizes each steal-back against general steals, but this time also using the size of the serial tasks to estimate the number of general steals.

\begin{theorem}
\label{ratio}
Define $N$ and $T_\infty'$ as in Theorem \ref{steal-back}, and let $X$ be any positive integer. Remove a total of at most $X$ serial tasks from consideration. (For example, it is a good idea to exclude the largest or the smallest serial tasks.) For each processor $i$, let $M_i$ denote the ratio between its largest and the smallest serial tasks after the removal. Let $M=\max_i M_i$. Then the total number of steal-back attempts is $O(N\min(M,T_\infty'))+T_\infty'X$.
\end{theorem}

\begin{proof}

There can be at most $T_\infty'X$ steal-backs performed on subtrees that include one of the $X$ serial tasks, since each subtree has height at most $T_\infty'$. 

Consider any other steal-back that processor $P_i$ performs on processor $P_j$. It is performed against a subtree that does not include one of the $X$ serial tasks. Therefore, it obtains at least $1/(M+1)$ of the total work in that subtree, leaving at most $M/(M+1)$ of the total work in $P_j$'s subtree. We amortize the steal-back against the general steal that $P_j$ performed on $P_i$ earlier. 

How many steal-backs can that general steal generate? We first assume that there are no general steals performed on $P_i$ or $P_j$ during the steal-backs. Then, $P_i$ can only steal back at most half of $P_j$'s work (since $P_j$ is working all the time, and thus will finish half of its work by the time $P_i$ steals half of its work).  To obtain the estimate, we solve for $K$ such that 
\[\left(\dfrac{M}{M+1}\right)^K=\dfrac12,\] 
and we obtain
\[K=\dfrac{\lg 2}{\lg(M+1)-\lg(M)}.\]
By integration, we have 
\[\int_M^{M+1}\dfrac{1}{M+1} dx < \int_M^{M+1}\dfrac{1}{x}dx < \int_M^{M+1}\dfrac{1}{M}dx,\] 
so that 
\[\dfrac{1}{M+1}<\ln(M+1)-\ln(M)<\dfrac{1}{M},\] or \[M<\dfrac{1}{\ln(M+1)-\ln(M)}<M+1.\] 
Since $\lg$ and $\ln$ are off each other by only a constant factor, $K$ grows as $O(M)$. This means that one random steal will be amortized against at most $O(M)$ steal-backs. Combined with the estimate involving $T_\infty$ from Theorem \ref{steal-back}, we have the desired bound, assuming that there are no general steals performed on $P_i$ or $P_j$ during these steal-backs.

Now we show that this last assumption is in fact unnecessary. That is, if there are general steals performed on $P_i$ or $P_j$ during these steal-backs, our estimate still holds. If a general steal is performed on $P_i$ after $P_i$ steals back from $P_j$, we amortize this steal-back against this general steal instead of against the general steal that $P_j$ made on $P_i$. Since each general steal can be amortized against in this way by at most one steal-back, our estimate holds. 

On the other hand, if a general steal is performed on $P_j$, then the steal-backs that $P_i$ has performed on $P_j$ become an even higher proportion of $P_j$'s work, and the remaining steal-backs proceed as usual. So our estimate also holds in this case.
\end{proof}

Applying Theorem \ref{ratio}, we may choose $O(P)$ serial tasks to exclude from the computation of $M$ without paying any extra ``penalty'', since the penalty $O(PT_\infty')$ is the same as the number of general steals. After we have excluded these serial tasks, if $M$ turns out to be constant, we obtain the desired $O(PT_\infty)$ bound on the number of steal-backs. The next theorem formalizes this fact.

\begin{theorem}
Define $N$ and $T_\infty'$ as in Theorem \ref{steal-back}, and remove any $O(P)$ serial tasks from consideration. For each processor $i$, let $M_i$ denote the ratio between its largest and the smallest serial tasks after the removal. Let $M=\max_i M_i$. Then the expected execution time on $P$ processors is $T_1/P+O(T_\infty\min(M,T_\infty'))$.
\end{theorem}

\begin{proof}
The amount of work is $T_1$, and Theorem \ref{ratio} gives a bound on the number of steal-back attempts in terms of the number of steal attempts. Since we know that the expected number of steal attempts is $O(PT_\infty)$, the expected number of steal-back attempts is $O(PT_\infty\min(M,T_\infty'))$. We add this to the amount of work and divide by $P$ to complete the proof.
\end{proof}

\begin{remark}
In the general case, it is not sufficient to amortize the steal-backs against the general steals. That is, there can be (asymptotically) more steal-backs than general steals, as is shown by the following example. 

Suppose that the adversary has control over the general steals. When there are $k$ owners left, the adversary picks one of them, say $P_i$. The other $k-1$ owners are stuck on a large serial task while $P_i$'s task is being completed. The $P-k$ free agents perform general steals so that $P_i$'s tree is split evenly (in terms of the number of serial tasks, not the actual amount of work) among the $P-k+1$ processors. Then $P_i$ finishes its work, while the other $P-k$ processors are stuck on a large serial task. $P_i$ performs repeated steal-backs on the $P-k$ processors until each of them is only down to its large serial task. Then they finish, and we are down to $k-1$ owners. In this case, $O(P^2T_\infty)$ steal-backs are performed, but only $O(P^2)$ general steals. 

In particular, it is not sufficient to use the bound on the number of general steals as a ``black box'' to bound the number of steal-backs. We still need to use the fact that the general steals are random.
\end{remark}


\section{Other Strategies}
\label{sec:localizedvariants}

In this section, we consider two variants of the localized work-stealing algorithm. The first variant, hashing, is designed to alleviate the problem of pile-up in the localized work-stealing algorithm. It assigns an equal probability in a steal-back to each owner that has work left. In the second variant, mugging, a steal-back takes all or almost all of the work of the processor being stolen from. A simple amortization argument yields an expected number of steals of $O(PT_\infty)$.

\subsection*{Hashing}

Intuitively, the way in which the general steals are set up in the localized work-stealing algorithm supports pile-up on certain processors' work. Indeed, if there are several processors working on processor $P_1$'s work, the next general steal is more likely to get $P_1$'s work, in turn further increasing the number of processors working on $P_1$'s work.

A possible modification of the general steal, which we call \textit{hashing}, operates as follows: first choose an owner uniformly at random among the owners who still has work left, then choose a processor that is working on that owner's work uniformly at random.

Loosely speaking, this modification helps in the critical path analysis both with regard to the general steals and to the steal-backs. Previously, if there are $k$ owners left, a general steal has a $\dfrac{k}{P}$ probability of hitting one of the $k$ remaining critical paths. Now, suppose there are $P_1,P_2,\ldots,P_k$ processors working on the $k$ owners' work, where $P_1+\ldots+P_k=P$. The probability of hitting one of the critical paths is 
\[\dfrac{1}{k}\left(\dfrac{1}{P_1}+\ldots+\dfrac{1}{P_k}\right)\geq \dfrac{k}{P}\] 
by the arithmetic-harmonic mean inequality \cite{Gwanyama04}. Also, the modified algorithm chooses the owner randomly, giving each owner an equal probability of being stolen from.

\subsection*{Mugging}

A possible modification of the steal-back, which we call \textit{mugging}, operates as follows: instead of $P_i$ taking only the top thread from $P_j$'s deque during a steal-back (i.e. half the tree), $P_i$ takes either (1) the whole deque, except for the thread that $P_j$ is working on; or (2) the whole deque, including the thread that $P_j$ is working on (in effect preempting $P_j$.) Figure \ref{fig:mugging} shows the processor of $P_j$ in each of the cases. 

Figure \ref{fig:mugging}(a) corresponds to the unmodified case, Figure \ref{fig:mugging}(b) to case (1), and Figure \ref{fig:mugging}(c) to case (2). The yellow threads are the ones that $P_i$ steals from $P_j$, while the white threads are the ones that $P_j$ is working on. In Figure \ref{fig:mugging}(c), the bottom thread is preempted by $P_i$'s steal.

In both modifications here, each general steal can generate at most one steal-back. Therefore, the expected number of steal-backs is $O(PT_\infty)$, and the expected number of total steals is also $O(PT_\infty)$.

\pgfdeclarepatternformonly{stripes}
{\pgfpointorigin}{\pgfpoint{1cm}{1cm}}
{\pgfpoint{1cm}{1cm}}
{
    \pgfpathmoveto{\pgfpoint{0cm}{0cm}}
    \pgfpathlineto{\pgfpoint{1cm}{1cm}}
    \pgfpathlineto{\pgfpoint{1cm}{0.5cm}}
    \pgfpathlineto{\pgfpoint{0.5cm}{0cm}}
    \pgfpathclose%
    \pgfusepath{fill}
    \pgfpathmoveto{\pgfpoint{0cm}{0.5cm}}
    \pgfpathlineto{\pgfpoint{0cm}{1cm}}
    \pgfpathlineto{\pgfpoint{0.5cm}{1cm}}
    \pgfpathclose%
    \pgfusepath{fill}
}

\begin{figure}
\centering

\begin{tabular}{c}
\begin{tikzpicture}[scale=0.7]

  \draw [fill=white,very thick] (0,0) rectangle (2,1);
  \draw [fill=lightgray,very thick] (0,1) rectangle (2,2);
  \draw [fill=lightgray,very thick] (0,2) rectangle (2,5);
  \draw [fill=lightgray,very thick] (0,5) rectangle (2,6);
  \draw [fill=yellow,very thick] (0,6) rectangle (2,7);
  \fill (1,2.75) circle (2pt);
  \fill (1,3.5) circle (2pt);
  \fill (1,4.25) circle (2pt);

\end{tikzpicture}
\\[\abovecaptionskip]
\small (a) Work stealing
\end{tabular}
\begin{tabular}{c}
\begin{tikzpicture}[scale=0.7]

  \draw [fill=white,very thick] (0,0) rectangle (2,1);
  \draw [fill=yellow,very thick] (0,1) rectangle (2,2);
  \draw [fill=yellow,very thick] (0,2) rectangle (2,5);
  \draw [fill=yellow,very thick] (0,5) rectangle (2,6);
  \draw [fill=yellow,very thick] (0,6) rectangle (2,7);
  \fill (1,2.75) circle (2pt);
  \fill (1,3.5) circle (2pt);
  \fill (1,4.25) circle (2pt);

\end{tikzpicture}
\\[\abovecaptionskip]
\small (b) Variant (1) of mugging
\end{tabular}
\begin{tabular}{c}
\begin{tikzpicture}[scale=0.7]

  \draw [fill=white,very thick,pattern=stripes, pattern color=yellow] (0,0) rectangle (2,1);
  \draw [fill=yellow,very thick] (0,1) rectangle (2,2);
  \draw [fill=yellow,very thick] (0,2) rectangle (2,5);
  \draw [fill=yellow,very thick] (0,5) rectangle (2,6);
  \draw [fill=yellow,very thick] (0,6) rectangle (2,7);
  \fill (1,2.75) circle (2pt);
  \fill (1,3.5) circle (2pt);
  \fill (1,4.25) circle (2pt);

\end{tikzpicture}
\\[\abovecaptionskip]
\small (c) Variant (2) of mugging
\end{tabular}

\caption{Deque of processor in variants of mugging} \label{fig:mugging}

\end{figure}
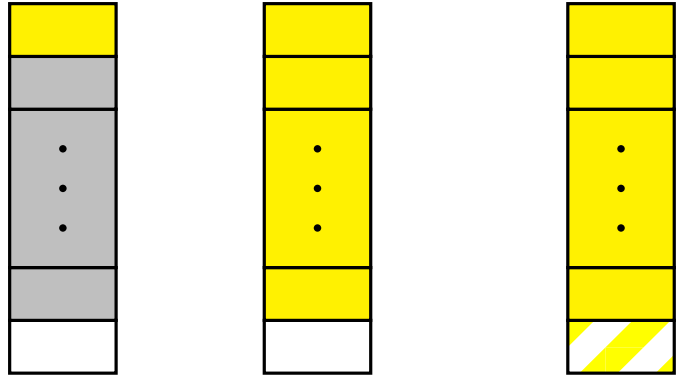

\section{Conclusion and Future Work}
\label{sec:conclusion}

In this paper, we have established running-time bounds on the localized work-stealing algorithm based on the delay-sequence argument and on amortization analysis. Here we suggest two possible directions for future work:

\begin{itemize}
\item This paper focuses on the setting in which the computation is modeled by binary trees. Can we achieve similar bounds for more general computational settings, e.g., one in which the computation is modeled by directed acyclic graphs (DAG)?
\item The hashing variant of the localized work-stealing algorithm (Section \ref{sec:localizedvariants}) is designed to counter the effect of pile-up on certain processors' work. What guarantees can we prove on the running time or the number of steals?
\end{itemize}




\end{document}